\documentclass[10pt]{article}
\usepackage[a4paper,left=3cm,right=3cm]{geometry}
\usepackage[utf8]{inputenc}
\usepackage[T1]{fontenc}
\usepackage{hyperref}
\usepackage[symbol]{footmisc}
\usepackage{authblk}
\usepackage{siunitx}
\usepackage{booktabs}
\usepackage{amsmath,amssymb,amsthm}
\usepackage{cases}
\usepackage{url}
\usepackage{bm}
\usepackage{graphicx}
\usepackage{enumerate}
\usepackage{algorithm}
\usepackage{mathtools}
\usepackage{todonotes}
\usepackage{scalerel}

\theoremstyle{plain}
\newtheorem{thm}{Theorem}[section]

\newtheorem{cor}[thm]{Corollary}

\theoremstyle{remark}
\newtheorem{rem}{Remark}[section]

\newcommand{\Econd}[2]{\mathbb{E}\left[\left.#1\right|#2\right]}

\newcommand{\E}[1]{\mathbb{E}\left[#1\right]}
\newcommand{\Ex}[2]{\mathbb{E}^{#1}\left[#2\right]}

\newcommand{\RR}{\mathbb{R}}

\newcommand{\QQ}{\mathbb{Q}}
\newcommand{\PP}{\mathbb{P}}

\newcommand{\EE}{\mathbb{E}}

\newcommand{\cF}{\mathcal{F}}
\newcommand{\cH}{\mathcal{H}}

\newcommand{\set}[1]{\left\{#1\right\}}

\newcommand{\norm}[1]{\left\lVert#1\right\rVert}

\newcommand{\sign}[1]{\mathrm{sign}\left(#1\right)}
\newcommand{\scal}[2]{\left\langle#1,#2\right\rangle}
\newcommand{\BSDelta}{\mathsf{Delta}}
\newcommand{\BSVega}{\mathsf{Vega}}
\newcommand{\BSd}{\mathsf{d}}
\newcommand{\MSHE}{\mathrm{MSHE}}
\newcommand\wh[1]{\hstretch{2}{\hat{\hstretch{.5}{#1}}}}
\newcommand{\pd}[2]{\frac{\partial #1}{\partial #2}}

\begin{document}
\pagestyle{empty}

\title{Bartlett's Delta revisited: Variance-optimal hedging in the lognormal SABR and in the rough Bergomi model}
\author{Martin Keller-Ressel}
\affil{TU Dresden, Institute for Mathematical Stochastics, Dresden, 01062, Germany}
\affil{martin.keller-ressel@tu-dresden.de}

\maketitle

\begin{abstract}We derive analytic expressions for the variance-optimal hedging strategy and its mean-square hedging error in the lognormal SABR and in the rough Bergomi model. In the SABR model, we show that the variance-optimal hedging strategy coincides with the Delta adjustment of Bartlett [Wilmott magazine 4/6 (2006)]. We show both mathematically and in simulation that the efficiency of the variance-optimal strategy (in comparison to simple Delta hedging) depends strongly on the leverage parameter $\rho$ and  -- in a weaker sense -- also on the roughness parameter $H$ of the model, and give a precise quantification of this dependency.
\end{abstract}

%\begin{acknowledgements}
%\end{acknowledgements}

\section{Introduction}\label{sec:intro}
Delta-hedging, as the unique strategy to eliminate all hedging risk in the Black-Scholes model, has been one of the pillars of classical finance. By using the option-implied Delta, Delta-hedging is easily adapted to models beyond Black-Scholes. However, it has been well understood that other risk factors, such as changes in volatility (Vega risk), also have to be taken into account in effective hedging strategies. Within the setting of local or stochastic volatility models, the derivation of such adjustments to Delta-hedging has been repeatedly considered in the literature. Crep\'ey, for example, proposes in \cite{crepey2004delta} an adjustment of the form
\begin{equation}\label{eq:adj_Crepey}
\theta_t^{Crepey} = \BSDelta_t + \BSVega_t \pd{\Sigma_t}{K},
\end{equation}
where $\pd{\Sigma_t}{K}$ is the derivative of implied volatility with respect to strike. Hull and White consider in \cite{hull2017optimal} an adjustment of the form
\begin{equation}\label{eq:HW}
\theta_t^{HW} = \BSDelta_t + \frac{\BSVega_t}{S_t \sqrt{T-t}} \left(a + b\,\BSDelta_t + c \,\BSDelta_t^2 \right),
\end{equation}
where the coefficients $a$, $b$ and $c$ are determined by regression. Importantly, they emphasize that their regression approach corresponds to \emph{minimizing the variance} of the hedging error. Here, we also follow this variance-optimality approach, which has originally been introduced and thoroughly explored in a general setting by F\"ollmer and Schweizer \cite{schweizer1984varianten, follmer1988hedging}. We assume that a description of the dynamic evolution of the option-implied volatility $\Sigma$ is available, and we obtain a general expression for the variance-minimizing hedging strategy and its mean-square error in Section~\ref{sec:bg}. We then focus on the application of these theoretical results to optimal hedging in the SABR model \cite{hagan2002managing} and in its extension to `rough volatility' (see \cite{gatheral2018volatility}), the rough Bergomi model introduced in \cite{bayer2016pricing}.\\
Already in \cite{hagan2002managing}, an adjusted Delta hedge for the SABR model is proposed by Hagan-Kumar-Lesniewski-Woodward (HKLW), taking the form
\begin{equation}\label{eq:HKLW_hedge}
\theta_t^{HKLW} = \BSDelta_t + \BSVega_t \frac{\partial \Sigma_t}{\partial S}.
\end{equation}
Aiming to improve the HLKW strategy, Bartlett proposes in \cite{bartlett2006hedging} the adjustment\footnote{All parameters refer to the SABR model as discussed in Section~\ref{sec:SABR} below.}
\begin{equation}\label{eq:Bartlett_hedge}
\theta_t^{Bartlett} = \BSDelta_t + \BSVega_t \left(\frac{\rho}{2 \eta} \frac{\partial \Sigma_t}{\partial \alpha} + \frac{\partial \Sigma_t}{\partial S}\right),
\end{equation}
which is subsequently discussed as `Bartlett's Delta' in \cite{hagan2017bartlett}. In Section~\ref{sec:SABR}, we show mathematically that Bartlett's hedge corresponds exactly to the variance-optimal hedge in the SABR model. Moreover, we derive analytic expressions for the mean-square hedging errors of the classic Delta-hedge, the HLKW-hedge and the Bartlett/variance-optimal hedge, allowing us to compare the strategies on a theoretical basis, in addition to their numerical evaluation. Subsequently, we show in Section~\ref{sec:rough} that our `dynamic implied volatility' framework for variance-optimal hedging can also be applied to rough stochastic volatility models, which are neither Markovian nor semi-martingales. Combining our framework with the implied-volatility approximations of Fukasawa and Gatheral \cite{fukasawa2022rough} for the rough Bergomi model \cite{bayer2016pricing}, we derive analytic expressions for the variance-optimal hedging strategy and its mean-square error. In particular, we discover an interesting interplay between the leverage parameter $\rho$ and the roughness parameter $H$ in determining the efficiency of the variance-optimal strategy. In section~\ref{sec:numerics} we report numerical results for the SABR and the rough Bergomi model, which confirm our theoretical findings. 

\section{Background}\label{sec:bg}
\subsection{Variance-optimal hedging}
Let $S = (S_t)_{t \geq 0}$ be the stock price and $C = (C_t)_{t \in [0,T]}$ be the price of a contingent claim with maturity $T$. For simplicity we set interest rates to zero and we assume that there exists a risk-neutral measure $\QQ$ under which both $S$ and $C$ are square-integrable martingales. Let $\theta$ be a strategy to hedge the payoff $C_T$. Then we can write 
\[C_t = w + \int_0^t \theta_u dS_u+ L_t,\]
where the terms on the right hand denote, respectively: the initial capital $w \in \RR$, the accumulated value of the hedging portfolio and the \emph{pathwise hedging error} $L$. Note the similarities to linear regression: We are regressing (the stochastic process) $C$ onto the hedging portfolios in $S$, with regression coefficient $\theta$, intercept $w$, and residual $L$. In a complete market we can find a perfect hedging strategy $\theta^\text{perf}$ such that the terminal pathwise hedging error $L_T$ is identical to zero and the unique initial capital needed is $w = \EE^\QQ(C_T)$. However, in an incomplete market model, such as the SABR model or other rough and non-rough stochastic volatility models, there is no perfect hedging strategy. Thus, we have to decide on another criterion of optimality to select a suitable strategy $\theta$. Here, we focus on the principle of \emph{variance-optimality}, introduced by \cite{schweizer1984varianten, follmer1988hedging}, which aims to find the initial capital $w$ and the variance-optimal strategy $\theta^{VO}$, which minimizes the risk-neutral\footnote{One can also start from a non-risk-neutral measure $\PP$ and minimize the MSHE under $\PP$. This leads to the rich subject of mean-variance hedging, cf. \cite{schweizer1992mean}, which we do not consider here.} mean-square hedging error (MSHE)
\begin{equation}\label{eq:MSHE}
\MSHE(\theta) = \Ex{\QQ}{L_T^2} = \Ex{\QQ}{\left(C_T -  w + \int_0^T \theta_u dS_u \right)^2}.
\end{equation}
The proper mathematical framework is to assume that $S$ belongs to the Hilbert space (cf. \cite[I.4a]{jacod2013limit})
\[\cH_2 = \set{M = (M_t)_{t \in [0,T]}: \text{$M$ is a martingale with } \E{M_T^2}< \infty}\]
of square integrable continuous martingales with norm $\norm{M}_2^2 = \E{M_T^2}$ and inner product $M \cdot N = \E{M_T N_T}$; and to restrict the hedging strategies to 
\[L_S^2 = \set{\theta \text{ adapted}: \E{\int_0^T \theta_t^2 d\scal{S}{S}_t} < \infty}.\] 
Finally, two martingales $M,N \in \cH_2$ are called \emph{orthogonal in the strong sense}, if their quadratic covariation vanishes, i.e., if $\scal{N}{M} = 0$. If $M_0 N_0 = 0$, then this strong orthogonality implies orthogonality in the Hilbert-space $\cH_2$.  
 The key observations are that 
\begin{itemize}
\item equation \eqref{eq:MSHE} is minimized by the orthogonal projection in $\cH_2$ of the martingale $C = (C_t)_{t \in [0,T]}$ onto the closed subspace spanned by the integrals $\int_0^t \theta_u dS_u$ where $\theta$ ranges through $L_S^2$; and that 
\item the residual  $L$ of this projection must be orthogonal to $S$; even in the strong sense of $\scal{L}{S} = 0$ (cf. \cite{kunita1967square}).
\end{itemize}
Strong orthogonality of $S$ and $L$ implies that 
\[0 = \scal{L}{S}_t = \scal{C}{S}_t + \int_0^t \theta_s d\scal{S}{S}_s\]
for all $t \in [0,T]$. Rearranging yields the 
variance-optimal strategy as
\[\theta_t^{VO} = \frac{d\scal{S}{C}_t}{d\scal{S}{S}_t},\]
where the right-hand side has to be read as the Radon-Nikodym derivative of the finite-variation processes $\scal{S}{C}_t$ and $\scal{S}{S}_t$. Note the similarity to linear regression, where the regression coefficient is obtained as the covariance between the dependent and the independent variable, divided by the variance of the independent variable.\\
Above, and in what follows, we make frequent use of the following calculation rules for quadratic covariations of continuous Ito processes $X,Y$ (see \cite[\S I.4d-e]{jacod2013limit})  of dimensions $m$ and $n$ respectively: 
\begin{enumerate}[{Q}-I.]
\item If $f \in C^2(\RR^m; \RR)$ and $g \in C^2(\RR^n; \RR)$, then $d\scal{f(X)}{g(Y)} = \nabla f(X)^\top d\scal{X}{Y} \nabla g(X)$, \label{eq:Q1}
\item If $\theta \in L_Y^2$, then $d\scal{\int \theta\,dY}{X} = \theta\,d\scal{Y}{X}$ \label{eq:Q2}
\item $\scal{X}{Y} = 0$ if $X$ or $Y$ has finite variation. \label{eq:Q3}
\end{enumerate}

\subsection{Variance-optimal hedging and implied volatility}
We restrict our attention to the case where $C$ is the price process of a call option with maturity $T$ and strike price $K$. To the call-price $C_t$ we can associate its (time-$t$) implied volatility $\Sigma_t$ and rewrite $C_t$ as 
\begin{equation}\label{eq:BS}
C_t = c_\text{BS}(T-t,S_t,\Sigma_t),
\end{equation}
where $c_\text{BS}$ is the Black-Scholes price in dependence on time-to-maturity, underlying and implied volatility. As a smooth function of the martingales $S$ and $C$, the process $\Sigma$ must be a semi-martingale. We introduce the following notation for the Black-Scholes Greeks and related quantities, given $\Sigma$:
\begin{align*}
\BSd_t^{\pm} &= \frac{\log(S_t/K)}{\Sigma_t \sqrt{T-t}} \pm \frac{\Sigma_t \sqrt{T-t}}{2}\\
\BSDelta_t &= \Phi(\BSd_t^+)\\
\BSVega_t &= S_t \phi(\BSd_t^+) \sqrt{T-t}
\end{align*}
Combining \eqref{eq:BS} with the variance-optimal hedging framework from above, we get a representation of the variance-optimal strategy in terms of Delta and Vega:

\begin{thm}\label{thm:VO_generic}
The variance-optimal hedging strategy for a call option $C$ with implied volatility process $\Sigma$ is
\begin{equation}\label{eq:VO_generic}
\theta_t^{VO} = \BSDelta_t   + \BSVega_t \frac{d\scal{\Sigma}{S}_t}{d\scal{S}{S}_t}
\end{equation}
with initial capital $w =\E{C_T}$. The mean-square hedging error of this strategy is 
\begin{equation}\label{eq:VO_MSHE_generic}
\MSHE(\theta^{VO}) = \E{\int_0^T \BSVega_t^2 dQ^\Sigma_t}
\end{equation}
where the `orthogonal vol-of-vol' process $Q^\Sigma$ is given by
\begin{equation}\label{eq:volovol}
dQ^\Sigma_t = d\scal{\Sigma}{\Sigma}_t - \frac{d\scal{\Sigma}{S}_t}{d\scal{S}{S}_t} d\scal{\Sigma}{S}_t.
\end{equation}
\end{thm}
\begin{rem}
We add some intuition to the nature of the orthogonal vol-of-vol process $Q^\Sigma$. First, define 
\[\Sigma^\perp := \Sigma - \frac{d\scal{\Sigma}{S}}{d\scal{S}{S}}S,\]
which is the component of $\Sigma$ which is orthogonal to $S$, i.e., we have
$d\scal{\Sigma^\perp}{S} = 0$. The process $Q^\Sigma$ can now be written as 
$dQ^\Sigma = d\scal{\Sigma^\perp}{\Sigma^\perp}$, hence `orthogonal vol-of-vol'.
\end{rem}
In addition, we can derive the following formula to evaluate the mean-square hedging error of \emph{any} (not necessarily variance-optimal) hedging strategy:
\begin{thm}\label{thm:MSHE}
Let $\theta$ be a hedging strategy for a call option $C$ with initial capital $w = \Ex{\QQ}{C_T}$ and implied volatility process $\Sigma$. Then the mean-square-hedging error of the strategy is given by
\begin{equation}\label{eq:MSHE_general}
\MSHE(\theta) = \E{\int_0^T \Big((\theta_t - \BSDelta_t)^2 d\scal{S}{S}_t - 2(\theta_t - \BSDelta_t) \BSVega_t d\scal{\Sigma}{S}_t + \BSVega_t^2 d\scal{\Sigma}{\Sigma}_t\Big) dt}
\end{equation}.
\end{thm}
Applying this formula to the simple Delta-hedge we obtain the following corollary:
\begin{cor}\label{cor:MSHE_Delta}
The mean-square hedging error of the Delta-Hedging strategy $\Delta_t := \BSDelta_t$
is
\begin{equation}\label{eq:Delta_MSHE_generic}
\MSHE(\Delta) = \E{\int_0^T \BSVega_t^2 d\scal{\Sigma}{\Sigma}},
\end{equation}
and the difference to the error of the variance-optimal strategy \eqref{eq:VO_MSHE_generic} is equal to 
\begin{equation}\label{eq:Diff_MSHE_generic}
\MSHE(\Delta)  - \MSHE(\theta^{VO})= \E{\int_0^T \BSVega_t^2 \left(\frac{d\scal{\Sigma}{S}_t}{d\scal{S}{S}_t} \right)^2 d\scal{S}{S}_t} \ge 0.
\end{equation}
\end{cor}
\begin{rem}
We note that the only difference between expressions \eqref{eq:VO_MSHE_generic} and \eqref{eq:Delta_MSHE_generic} for the hedging error of the variance-optimal and the Delta-hedging strategy is that the integral is taken with respect to \emph{orthogonal} vol-of-vol in the former, and with respect to ordinary vol-of-vol in the latter. 
\end{rem}

\begin{proof}[Proof of Thms. \ref{thm:VO_generic}, \ref{thm:MSHE} and Cor. \ref{cor:MSHE_Delta}]
Using property \textbf{Q-I} from above, we have
\begin{align*}
\theta_t^{VO} &= \frac{d\scal{S}{C}_t}{\scal{S}{S}_t} = \pd{c}{S}(T-t,S_t,\Sigma_t) \frac{d\scal{S}{S}_t}{d\scal{S}{S}_t} + \pd{c}{\Sigma}(T-t,S_t,\Sigma_t) \frac{d\scal{\Sigma}{S}_t}{d\scal{S}{S}_t} = \\
&= \BSDelta_t   + \BSVega_t \frac{d\scal{\Sigma}{S}_t}{d\scal{S}{S}_t}.
\end{align*}
To compute the MSHE of an arbitrary strategy $\theta \in L_S^2$, note that
\[\E{L_T^2} = \E{\scal{L}{L}_T} = \E{\scal{C - w - \int_0 \theta dS}{C - w - \int_0 \theta dS}_T}\]
and calculate, using property \textbf{Q-II}, 
\begin{align}
d\scal{L}{L}_t &= d\scal{C}{C}_t - 2 \scal{C}{\int \theta dS}_t + \scal{\int \theta dS}{\int \theta dS}_t = \notag \\
&= d\scal{C}{C}_t - 2 \theta_t \scal{C}{S}_t + \theta_t^2 \scal{S}{S}_t \label{eq:L_inter}.
\end{align}
Using property \textbf{Q-I} again, we obtain
\begin{align*}
d\scal{L}{L}_t &= \BSDelta_t^2 d\scal{S}{S}_t + 2\,\BSDelta_t \BSVega_t d\scal{S}{\Sigma}_t + \BSVega_t^2 d\scal{\Sigma}{\Sigma}_t - \\
&\phantom{=}-2 \theta_t \BSDelta_t d\scal{S}{S}_t - 2 \theta_t \BSVega_t d\scal{S}{\Sigma}_t + \theta_t^2 d\scal{\Sigma}{\Sigma}_t  = \\
&= (\theta_t - \BSDelta_t)^2 d\scal{S}{S}_t - 2(\theta_t - \BSDelta_t) \BSVega_t d\scal{\Sigma}{S}_t + \BSVega_t^2 d\scal{\Sigma}{\Sigma}_t,
\end{align*}
which gives \eqref{eq:MSHE_general}. Inserting the variance-optimal strategy $\theta^{VO}$ yields
\begin{align*}
d\scal{L}{L}_t &= \BSVega_t^2 \Big( \left(\frac{d\scal{\Sigma}{S}_t}{d\scal{S}{S}_t}\right)^2 d\scal{S}{S}_t - 2 
\frac{d\scal{\Sigma}{S}_t}{d\scal{S}{S}_t} d\scal{\Sigma}{S}_t + d\scal{\Sigma}{\Sigma}_t \Big) = \\
&= \BSVega_t^2 \Big(d\scal{\Sigma}{\Sigma}_t - \frac{d\scal{\Sigma}{S}_t}{d\scal{S}{S}_t} d\scal{\Sigma}{S}_t\Big)
\end{align*}
which gives \eqref{eq:VO_MSHE_generic}. Inserting the Delta-hedging strategy, on the other hand, yields \eqref{eq:Delta_MSHE_generic}.
\end{proof}

To compute the variance-optimal strategy from Theorem~\ref{thm:VO_generic}, we need a tractable description of the dynamic implied volatility process $\Sigma = (\Sigma_t)_{t \in [0,T]}$, which is rarely available in stochastic volatility models. However, in certain models, such as the SABR and the rough Bergomi model, accurate \emph{dynamic approximations} $\hat \Sigma_t$ of $\Sigma_t$ are available, due to \cite{hagan2002managing, balland2006forward, fukasawa2022rough}. This is our key to obtaining explicit approximate variance-optional strategies for these models in the following sections.

\section{The SABR model}\label{sec:SABR}
We consider the SABR model of \cite{hagan2002managing} in its conditionally lognormal form, i.e., with $\beta = 1$, which takes the form
\begin{align*}
dS_t &= S_t \alpha_t dB_t\\
d\alpha_t &= \frac{\eta}{2}\alpha_t dW_t
\end{align*}
where $d\scal{B}{W}_t = \rho dt$ with $\rho \in [-1,1]$. We parameterize vol-of-vol by $\tfrac{\eta}{2}$ for consistency with the rough Bergomi model, as discussed in Section~\ref{sec:rough}. An asymptotically arbitrage-free approximation $\hat \Sigma$ of $\Sigma$ in the lognormal SABR model is given by \cite{fukasawa2022rough} (see also \cite{balland2006forward}) as 
\begin{equation}\label{eq:AAA_SABR}
\hat \Sigma_t = \alpha_t f(Y_t), \qquad Y_t = \frac{\eta}{\alpha_t} \log \left(K/S_t\right),
\end{equation}
where $f$ is given by the famous SABR formula of Hagan et al.\ \cite{hagan2002managing}: 
\begin{equation}\label{eq:f}
f(y) = \frac{y}{g(y)}, \qquad g(y) = - 2 \log\left(\frac{\sqrt{1 + \rho y  + y^2/4} - \rho - y/2}{1 - \rho}\right).
\end{equation}
We denote by $\wh \BSDelta_t$, $\wh \BSVega_t$, $\wh \BSd^\pm$, etc. the Black-Scholes Greeks (and related quantities) evaluated at the approximation $\hat \Sigma$ rather than the exact implied volatility $\Sigma$. 

\subsection{Bartlett's Delta is the variance-optimal strategy}
We consider and compare the following three hedging strategies for the SABR model: 
\begin{itemize}
\item The (classic) Delta hedging strategy $\Delta_t = \Phi(\wh \BSd^+_t)$,  
which uses the Black-Scholes-Delta evaluated at the SABR-implied volatility;
\item The Hagan-Kumar-Lesniewski-Woodward-(HKLW) adjusted Delta strategy \eqref{eq:HKLW_hedge} proposed in \cite{hagan2002managing};
\item Bartlett's adjusted Delta strategy \eqref{eq:Bartlett_hedge} which was proposed in \cite{bartlett2006hedging} (see also \cite{hagan2017bartlett}) as an improvement of the HKLW strategy.
\end{itemize}

Our first result shows that Bartlett's strategy is variance-optimal (up to the approximation error induced by the approximation $\hat \Sigma \approx \Sigma$) , i.e., no other hedging strategy can obtain a smaller hedging error in the mean-square sense.
\begin{thm}\label{thm:SABR}The approximate variance-optimal strategy $\theta^{AVO}$ for the SABR model, obtained by substituting $\hat \Sigma$ for $\Sigma$ in \eqref{eq:VO_generic}, coincides with Bartlett's adjusted Delta strategy \eqref{eq:Bartlett_hedge} and is given by
%\begin{align}
%\theta_t^{AVO} &= \BSDelta_t + \frac{\eta}{2} \frac{\BSVega_t}{S_t} \left[\rho ( f(Y_t) - Y_t f'(Y_t)) - 2 f'(Y_t) \frac{\alpha_t}{U_t}\right] \approx \\
%&\approx \BSDelta_t + \frac{\eta}{2}\phi(d_1) \sqrt{T-t} \left[\rho ( f(Y_t) - Y_t f'(Y_t)) - 2 f'(Y_t) \right].
%\end{align}
\begin{align}
\theta_t^{AVO} &= \theta_t^{Bartlett} = \wh \BSDelta_t + \frac{\eta}{2S_t} \wh \BSVega_t \Big(\rho F_1(Y_t)  + F_2(Y_t)\big) =\notag\\
&= \Phi(\wh \BSd_t^+) + \frac{\eta}{2}\phi(\wh \BSd_t^+) \sqrt{T-t} \big(\rho F_1(Y_t)  + F_2(Y_t)\big),\label{eq:AVO_SABR}
\end{align}
where
\begin{align*}
F_1(y) &= F_1(y;\rho) = f(y) - y f'(y)\\
F_2(y) &=F_2(y;\rho) = -2 f'(y)
\end{align*}
with $f$ given by \eqref{eq:f}.
\end{thm}
\begin{rem}
\begin{enumerate}[(a)]
\item Note that $F_1(Y_t) = \pd{\Sigma_t}{\alpha}$ and $F_2(Y_t) = \frac{2S_t}{\eta}\pd{\Sigma_t}{S}$, hence \eqref{eq:AVO_SABR} is the same as Bartlett's Delta \eqref{eq:Bartlett_hedge}.
\item For implementations, the following representations of $F_1$ and $F_2$ are useful:
\begin{align*}
F_1(y) &= f(y)^2 / \sqrt{1 + \rho y  + y^2/4}\\
F_2(y) &= \tfrac{2}{y} \left(F_1(y) - f(y)\right).
\end{align*}
\end{enumerate}
\end{rem}

\begin{thm}\label{thm:SABR_MSHE}
The approximate mean-squared hedging errors for the Bartlett/variance-optimal strategy, the Hagan-Kumar-Lesniewski-Woodward strategy and the simple Delta strategy are 
\begin{align}\label{eq:AVO_MSHE_SABR}
\wh \MSHE(\theta^{AVO}) &= \frac{\eta^2}{4} ( 1- \rho^2) \int_0^t  \wh \BSVega_s^2\;\alpha_s^2 \;F_1(Y_s)^2 ds,\\
\wh \MSHE(\theta^{HKLW})  &= \frac{\eta^2}{4} \int_0^t  \wh \BSVega_s^2\;\alpha_s^2 F_1(Y_s)^2 ds,\\
\wh \MSHE(\Delta)  &= \frac{\eta^2}{4} \int_0^t  \wh \BSVega_s^2\;\alpha_s^2 (F_1(Y_s)^2 + 2 \rho F_1(Y_s) F_2(Y_s) + F_2(Y_s)^2) ds.
\end{align}
\end{thm}
\begin{rem}
\begin{enumerate}[(a)]
\item It can be seen that in the `complete market limit' $|\rho| \to 1$, the MSHE vanishes only for the Bartlett/variance-optimal strategy. The other two strategies are not able to fully exploit the correlation of stock price and stochastic variance. 
\item The formulas for the MSHE are calculated `within-approximation', i.e. the approximation $\Sigma \approx \wh \Sigma$ is used to both calculate the strategy and to evaluate its error. A more honest evaluation would be to calculate the strategy using $\hat \Sigma$ and to evaluate the error using $\Sigma$. We expect the error approximations to be biased towards zero (that is, a little too optimistic) compared to this honest evaluation.
\end{enumerate}
\end{rem}

Comparing the hedging errors of the strategies in Thm.~\ref{thm:SABR_MSHE}, we can compute the difference of the mean-squared hedging error of the simple Delta strategy and the Bartlett/variance-optimal strategy as
\begin{equation}\label{eq:difference_SABR}
\wh \MSHE(\Delta) - \wh\MSHE(\theta^{Bartlett}) = \frac{\eta^2}{4} \int_0^t  \wh \BSVega_s^2\;\alpha_s^2 (\rho F_1(Y_s) + F_2(Y_s)) ^2 ds \ge 0.
\end{equation}
Moreover, it is easy to see that
\[\wh \MSHE(\theta^{HKLW}) = \frac{\wh \MSHE(\theta^{Bartlett})}{1 - \rho^2},\]
i.e. the error of the HKLW-strategy is always larger than the error of the Bartlett/variance-optimal strategy by a factor of $1/(1 - \rho^2)$, independent of all other model parameters. 
Finally, taking the difference of HKLW and the simple Delta strategy yields
\begin{equation}\label{eq:Delta_vs_HKLW}
\wh\MSHE(\Delta) -  \wh \MSHE(\theta^{HKLW})  = \frac{\eta^2}{4} \int_0^t  \wh \BSVega_s^2\;\alpha_s^2 \left(F_2^2(Y_t) + 2 \rho F_1(Y_s) F_2(Y_s) \right)  ds.
\end{equation}
For $\rho \approx 0$, the HKLW strategy improves upon the simple Delta hedge, which was the original intention of Hagan et al. For larger $|\rho|$, the conclusion from \eqref{eq:Delta_vs_HKLW} becomes unclear; the first-order analysis below shows that we should expect  the simple Delta hedge to have the smaller error.

\begin{proof}[Proof of Thms.~\ref{thm:SABR} and \ref{thm:SABR_MSHE}]
Applying Ito's formula to \eqref{eq:AAA_SABR} we obtain
\begin{align*}
d \hat \Sigma_t &= \pd{\wh \Sigma_t}{\alpha} d\alpha_t + \pd{\wh \Sigma_t}{S} dS_t + \textit{drift} = \frac{\eta}{2} \alpha_t \left(f(Y_t) - f'(Y_t) Y_t\right) dW_t - \eta \alpha_t f'(Y_t) dB_t + \textit{drift} = \\&= \frac{\eta}{2} \alpha_t \left(F_1(Y_t)dW_t + F_2(Y_t)dB_t\right) + \textit{drift}.
\end{align*}
Hence, we have
\[d \scal{\hat \Sigma}{S}_t = \frac{\eta}{2}\alpha_t^2 S_t \left(\rho F_1(Y_t)  + F_2(Y_t)\right) dt.\]
Together with 
\[d \scal{S}{S}_t = S_t^2 \alpha_t^2 dt,\]
and inserting into \eqref{eq:VO_generic} we obtain the approximate variance-optimal strategy \eqref{eq:AVO_SABR}; coinciding with \eqref{eq:Bartlett_hedge}. For the mean-square hedging error, we calculate
\[d \scal{\hat \Sigma}{\hat \Sigma}_t = \frac{\eta^2}{4}\alpha_t^2  \left(F_1(Y_t)^2 + 2 \rho F_1(Y_t) F_2(Y_t)  + F_2(Y_t)^2\right) dt.\]
Inserting into \eqref{eq:volovol} we obtain the orthogonal vol-of-vol process
\[dQ^{\hat \Sigma} = \frac{\eta^2}{4} \alpha_t^2 \left\{(F_1(Y_t)^2 + 2\rho F_1(Y_t)F_2(Y_t) + F_2(Y_t)) - (\rho F_1(Y_t) + F_2(Y_t))^2\right\} = \frac{\eta^2}{4} \alpha_t^2 (1 - \rho^2) F_1(Y_t)^2;\]
from Thm.~\ref{thm:VO_generic} we obtain \eqref{eq:AVO_MSHE_SABR}. The expressions for the mean-square hedging errors of the HKLW- and the simple Delta strategy follow by applying Thm.~\ref{thm:MSHE}.
\end{proof}

\subsection{First-order analysis}
We simplify the comparison of the different hedging strategies, using the first-order approximation 
\[f(y) \approx 1 + \frac{\rho}{4} y,\]
of the SABR implied volatility from \cite{fukasawa2022rough}. Under this approximation
\begin{equation}
F_1(y) \approx 1, \qquad F_2(y) \approx -\frac{\rho}{2},
\end{equation}
and the Bartlett/variance-optimal strategy is approximated by 
\[\theta_t^{Bartlett} \approx \wh \BSDelta_t + \frac{\rho \eta}{4S_t}  \wh \BSVega_t.\]
The mean-squared hedging errors of the Bartlett/variance-optimal, the HKLW- and the Delta strategy become
\begin{align}\label{eq:SABR_approx1}
\wh \MSHE(\theta^{Bartlett}) &\approx  \frac{\eta^2}{4} ( 1- \rho^2) \int_0^t  \wh \BSVega_s^2\;\alpha_s^2 ds,\\
\wh \MSHE(\theta^{HKLW}) &\approx  \frac{\eta^2}{4} \int_0^t  \wh \BSVega_s^2\;\alpha_s^2 ds,\label{eq:SABR_approx2}\\
\wh \MSHE(\Delta) &\approx  \frac{\eta^2}{4} \left( 1- \frac{3\rho^2}{4}\right) \int_0^t  \wh \BSVega_s^2\;\alpha_s^2 ds,\label{eq:SABR_approx3}
\end{align}
such that our first-order analysis suggests
\begin{equation}\label{eq:ordering}
\wh \MSHE(\theta^{Bartlett}) \le \wh \MSHE(\Delta) \lesssim \wh \MSHE(\theta^{HKLW}).
\end{equation}
Note, however, that the non-asymptotic result \eqref{eq:Delta_vs_HKLW} shows an advantage of the HKLW-strategy over the Delta hedge for $\rho \approx 0$, which gets lost in the first-order approximation.\\
An interesting quantity to analyze is the \emph{relative reduction} in root-mean-square error of the Bartlett/variance-optimal strategy in comparison to the Delta hedge. Using \eqref{eq:SABR_approx1} and \eqref{eq:SABR_approx3} we obtain
\begin{equation}\label{eq:relred_SABR}
\textsf{RelRed}_\rho = \frac{\sqrt{\MSHE(\Delta)} - \sqrt{\MSHE(\theta^{Bartlett}) }}{\sqrt{\MSHE(\Delta)}} = 1 - 2 \sqrt{\frac{1 - \rho^2}{4 - 3\rho^2}}.
\end{equation}
Note that this expression is independent of the option's strike and its time-to-maturity, but loses accuracy for ITM/OTM options. The quantity is shown as a function of $\rho$ in Figure~\ref{fig:relred_theory}. It is flat around zero and only starts to rise steeply in vicinity of the endpoints $\rho = \pm 1$. This suggest that the advantage of the Bartlett/variance-optimal hedge over the Delta hedge only starts to manifest in models with a leverage $\rho$ close to $\pm 1$. A half-way reduction of the root-mean-squared error, for example, is achieved only at $\rho = \pm \sqrt{\frac{12}{13}} \approx \pm 0.96$.

\begin{figure}
\centering
\includegraphics[width=0.75\textwidth]{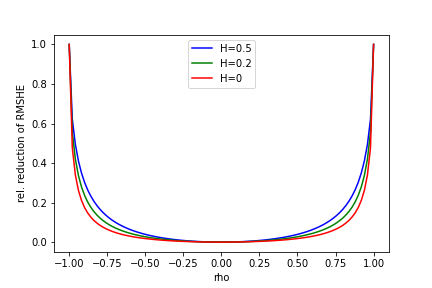} \label{fig:relred_theory}
\caption{\textbf{Relative Reduction in hedging error - first order approximation.} The first-order approximation \eqref{eq:relred_rough} (see also \eqref{eq:relred_SABR}) of the relative reduction in root-mean-squared hedging error for the variance-optimal strategy compared to the simple Delta hedge is shown in dependence on correlation $\rho$ and roughness parameter $H$.}
\end{figure}

\section{The rough Bergomi model}\label{sec:rough}
We consider the rough Bergomi model, introduced in \cite{bayer2016pricing} to generalize Bergomi's model \cite{bergomi2008smile}, with price process $S$ and forward variance $\xi_t(s) = \Econd{\alpha^2_s}{\cF_t}$ given by 
\begin{align*}
dS_t &= S_t \alpha_t dB_t\\
d\xi_t(s) &= \xi_t(s) \kappa(s-t)dW_t˛\quad s > t.
\end{align*}
Here, $d\scal{B}{W}_t = \rho dt$ with $\rho \in [-1,1]$; $\kappa$ is the power-law kernel $\kappa(r) = \eta \sqrt{2H} r^{H - 1/2}$ with $H \in [0,1/2]$, and $\alpha_t = \sqrt{\xi_t(t)}$. Note that the Hurst parameter $H$ controls the roughness of the volatility process $\alpha$ and the kernel $\kappa$ reproduces the power-law behavior of the volatility skew, cf. \cite{gatheral2018volatility}.
An asymptotically arbitrage-free approximation $\hat \Sigma$ of $\Sigma$ is given by \cite{fukasawa2022rough} as 
\begin{equation}\label{eq:AAA_Bergomi}
\hat \Sigma_t = U_t f(Y_t), \qquad Y_t = \frac{\kappa(T-t)}{U_t} \log \left(K/S_t\right),
\end{equation}
where 
\[U_t = \sqrt{\frac{1}{T-t}\int_t^T \xi_t(s)ds}\]
and $f$ is the solution of the ODE
\begin{equation}\label{eq:Bergomi_ODE}
\left(1 - \frac{yf'(y)}{f(y)}\right) \left(1- 2\rho \frac{y}{2H+1}  + \left(\frac{y}{2H+1}\right)^2 \right) = f(y)^2 \left(1 - (1- 2H)\frac{yf'(y)}{f(y)}\right), \qquad f(0) = 1.
\end{equation}
We will also need the auxilliary process
\[R_t = \frac{\int_t^T \kappa (s - t) \xi_t(s) ds}{\kappa(T-t) \int_t^T \xi_t(s) ds}\]
and use the approximations
\begin{equation}\label{eq:approx}
\frac{U_t}{\alpha_t} \approx 1, \quad \text{and} \quad R_t \approx \frac{1}{H + 1/2},
\end{equation}
which become exact in the limit $t \to T$, see \cite{fukasawa2022rough}. Note that they are also exact in the $H \to \tfrac{1}{2}$ limit, i.e., in the SABR model, see Sec.~\ref{sec:SABR}.

\subsection{Variance-Optimal Hedging}
\begin{thm}\label{thm:rough}The approximate variance-optimal strategy $\theta^{AVO}$ for the rough Bergomi model, obtained by substituting $\hat \Sigma$ for $\Sigma$ in \eqref{eq:VO_generic} and using the approximations \eqref{eq:approx}, is given by
%\begin{align}
%\theta_t^{AVO} &= \BSDelta_t + \frac{\eta}{2} \frac{\BSVega_t}{S_t} \left[\rho ( f(Y_t) - Y_t f'(Y_t)) - 2 f'(Y_t) \frac{\alpha_t}{U_t}\right] \approx \\
%&\approx \BSDelta_t + \frac{\eta}{2}\phi(d_1) \sqrt{T-t} \left[\rho ( f(Y_t) - Y_t f'(Y_t)) - 2 f'(Y_t) \right].
%\end{align}
\begin{align}
\theta_t^{AVO} &= \wh \BSDelta_t + \frac{\kappa(T-t)}{2} \frac{\wh \BSVega_t}{S_t} \Big(\frac{\rho}{H + 1/2} F_1(Y_t)  + F_2(Y_t)\Big) =\notag\\
&= \Phi(\wh \BSd_t^+) + \frac{\eta}{2}\phi(\wh \BSd_t^+) \sqrt{2H} (T-t)^H \Big(\frac{\rho}{H + 1/2} F_1(Y_t)  + F_2(Y_t)\Big),\label{eq:AVO_rough}
\end{align}
where
\begin{align*}
F_1(y) &= F_1(y;\rho,H) = f(y) - y f'(y)\\
F_2(y) &= F_2(y;\rho,H) = -2 f'(y),
\end{align*}
and $f$ the solution of \eqref{eq:Bergomi_ODE}. The approximate mean-squared hedging error, obtained by substituting $\hat \Sigma$ for $\Sigma$ in \eqref{eq:VO_MSHE_generic} and using the approximations \eqref{eq:approx} is given by
\begin{equation}\label{eq:AVO_MSHE_rough}
\wh \MSHE(\theta^{AVO}) = \frac{(1- \rho^2)}{(H + 1/2)^2} \int_0^T  2H (T-s)^{2H - 1} \wh \BSVega_s^2  \alpha_s^2 F_1(Y_s)^2 ds.
\end{equation}
\end{thm}
\begin{rem}Explicit expressions for the solution $f$ of \eqref{eq:Bergomi_ODE} (and hence of $F_1(y)$ and $F_2(y)$) are only available for $H = 0$ and $H = \tfrac{1}{2}$, but not for intermediate values. For $H \in (0,\tfrac{1}{2})$, we use the approximation of \cite{fukasawa2022rough}; see Appendix~\ref{app}.
\end{rem}

%\begin{thm}The approximate variance-optimal strategy $\theta^{AVO}$ for the rough Bergomi model, obtained by substituting $\hat \Sigma$ for $\Sigma$ in \eqref{eq:VO_generic} and using the approximations \eqref{eq:approx}, is given by
%%\begin{align}
%%\theta_t^{AVO} &= \BSDelta_t + \frac{\eta}{2} \frac{\BSVega_t}{S_t} \left[\rho ( f(Y_t) - Y_t f'(Y_t)) - 2 f'(Y_t) \frac{\alpha_t}{U_t}\right] \approx \\
%%&\approx \BSDelta_t + \frac{\eta}{2}\phi(d_1) \sqrt{T-t} \left[\rho ( f(Y_t) - Y_t f'(Y_t)) - 2 f'(Y_t) \right].
%%\end{align}
%\begin{align}
%\theta_t^{AVO} &= \wh \BSDelta_t + \frac{\kappa(T-t)}{2S_t} \wh \BSVega_t \Big(\rho R_t \frac{U_t}{\alpha_t} F_1(Y_t)  + F_2(Y_t)\big) =\notag\\
%&= \Phi(\wh \BSd_t^+) + \frac{\eta}{2}\phi(\wh \BSd_t^+) \sqrt{2H} (T-t)^H \big(\rho R_t \frac{U_t}{\alpha_t} F_1(Y_t)  + F_2(Y_t)\big),\label{eq:AVO}
%\end{align}
% The approximate mean-squared hedging error, obtained by substituting $\hat \Sigma$ for $\Sigma$ in \eqref{eq:VO_MSHE_generic}, is given by
%\begin{equation}\label{eq:AVO_MSHE}
%\wh \MSHE(\theta^{AVO}) = 2 H ( 1- \rho^2) \int_0^T  (T-s)^{2H - 1} \wh \BSVega_s^2 R_s^2 U_s^2 F_1(Y_s)^2 ds.
%\end{equation}
%\end{thm}

\begin{cor}\label{cor:Delta_rough}Let $\Delta_t = \wh \BSDelta_t$ denote the Delta-hedging strategy (using the approximate implied volatility $\hat \Sigma$). 
The approximate mean-square hedging error of this strategy is
\begin{equation}
\wh \MSHE(\Delta)  = \frac{\eta^2}{4} \int_0^t  \BSVega_s^2\;\alpha_s^2 \left(\left(\frac{F_1(Y_s)}{H+\tfrac{1}{2}}\right)^2 + 2 \rho \frac{F_1(Y_s)}{H + \tfrac{1}{2}} F_2(Y_s) + F_2(Y_s)^2\right) ds
\end{equation}
and its difference to the error of the approximate variance-optimal strategy is given by
\begin{equation}
\wh \MSHE(\Delta) - \wh \MSHE(\theta^{AVO}) = \frac{\eta^2}{4} \int_0^t  \BSVega_s^2\;\alpha_s^2 \left(\rho \frac{F_1(Y_s)}{H + \tfrac{1}{2}} + F_2(Y_s)\right) ^2 ds \ge 0.
\end{equation}
\end{cor}

\begin{proof}[Proof of Thm.~\ref{thm:rough} and Cor.~\ref{cor:Delta_rough}]
From \cite{fukasawa2022rough} we know that 
\[d \hat \Sigma_t = \frac{\kappa(T-t)}{2} \left(R_t U_t F_1(Y_t)dW_t + \alpha_t F_2(Y_t)dB_t\right) + \textit{drift}.\]
Hence, we have
\[d \scal{\hat \Sigma}{S}_t = \frac{\kappa(T-t)}{2}\alpha_t S_t \left(\rho R_t U_t F_1(Y_t)  + \alpha_t F_2(Y_t)\right) dt.\]
Together with 
\[d \scal{S}{S}_t = S_t^2 \alpha_t^2 dt,\]
and inserting into \eqref{eq:VO_generic} we obtain the approximate variance-optimal strategy \eqref{eq:AVO_rough}.
For the mean-square hedging error, we calculate
\[d \scal{\hat \Sigma}{\hat \Sigma}_t = \frac{\eta^2}{4}\left(R_t^2 U_t^2 F_1(Y_t)^2 + 2 \rho R_t U_t \alpha_t F_1(Y_t) F_2(Y_t)  + \alpha_t^2 F_2(Y_t)^2\right) dt.\]
Inserting into \eqref{eq:volovol} we obtain the orthogonal vol-of-vol process
\[dQ^{\hat \Sigma} = \frac{\kappa(T-t)^2}{4} (1 - \rho^2) R_t^2 U_t^2 F_1(Y_t)^2.\]
Together with \eqref{eq:VO_MSHE_generic} and using the approximations \eqref{eq:approx} we obtain \eqref{eq:AVO_MSHE_rough}. Cor.~\ref{cor:Delta_rough} now follows by an application of Cor.~\ref{cor:MSHE_Delta}.
\end{proof}

\subsection{First-order analysis}
We use the first-order approximation 
\[f(y) \approx 1 + \frac{a_H}{2} \rho y, \qquad a_H = \frac{1}{(H+\frac{1}{2})(H + \frac{3}{2})}\]
of the rough Bergomi implied volatility, as given in \cite{fukasawa2022rough}. Under this approximation
\begin{equation}
F_1(y) \approx 1, \qquad F_2(y) \approx -a_H \rho.
\end{equation}
Using these approximations, the error of the approximate variance-optimal and the Delta strategy become
\begin{align}\label{eq:Bergomi_approx}
\wh \MSHE(\theta^{AVO}) &\approx  \frac{( 1- \rho^2)}{(H + \tfrac{1}{2})^2}  \int_0^t  \kappa(t-s)^2 \wh \BSVega_s^2\;\alpha_s^2 ds, \\
\wh \MSHE(\Delta) &\approx  \frac{1}{(H + \tfrac{1}{2})^2} \left( 1- 2\rho^2\frac{H+1}{(H + \tfrac{3}{2})^2}\right) \int_0^t  \kappa(t-s)^2 \wh \BSVega_s^2\;\alpha_s^2 ds.
\end{align}
The ratio of these quantities evaluates to 
\[r = \frac{\wh \MSHE(\theta^{AVO})}{\wh \MSHE(\Delta)} = \Big(1 - \rho^2\Big)/\left(1 - 2\rho^2\frac{H+1}{(H + \tfrac{3}{2})^2}\right).\]
Taking derivatives, it is easy to see that this ratio is decreasing in $H$ on $(0,\tfrac{1}{2})$. Consequently, the relative reduction of the hedging error is given by 
\begin{equation}\label{eq:relred_rough}
\textsf{RelRed}_{\rho, H} = \frac{\sqrt{\MSHE(\Delta)} - \sqrt{\MSHE(\theta^{AVO})}}{\sqrt{\MSHE(\Delta)}} = 1 -  (H + \tfrac{3}{2}) \sqrt{\frac{1 - \rho^2}{(H + \tfrac{3}{2})^2 - 2(H + 1)\rho^2}}
\end{equation}
and is increasing in $H$. We conclude that the effectiveness of variance-optimal hedging, relative to Delta hedging, is largest in the SABR case and smallest in very rough models with $H$ close to zero. A plot of $\textsf{RelRed}_{\rho, H}$ in dependency of $\rho$ and for different $H$ is shown in Figure~\ref{fig:relred_theory}. While the effect of $H$ is visible, the influence of $\rho$ dominates, and is qualitatively similar for all $H \in [0,\tfrac{1}{2}]$. The boundary case $H = 0$ can also be simplified; for $H =0$ the relative error reduction evaluates to
\[\textsf{RelRed}_{\rho,0} = 1 -  3 \sqrt{\frac{1 - \rho^2}{9 - 8\rho^2}};\]
compare with \eqref{eq:relred_SABR}. A half-way reduction of the root-mean-squared error in the case $H = 0$, is achieved only at $\rho = \pm \sqrt{\frac{27}{28}} \approx \pm 0.98$.

\section{Numerical Results}\label{sec:numerics}
\subsection{Method}
To verify our results by simulation, we have implemented the SABR model and the calculation of the variance-optimal strategies in Python. For the simulation of the rough Bergomi model, we use the turbocharged Monte-Carlo scheme of \cite{mccrickerd2018turbocharging}, publicly available at \url{http://https://github.com/ryanmccrickerd/rough_bergomi}. As parameters for the SABR/rough Bergomi model we use
\[\eta = 0.5 \quad \text{and} \quad \alpha_0 = 0.4\]
and a flat term structure of initial forward variance is assumed in the rough Bergomi case. The parameters $\rho$ and $H$ are varied over a grid of 
\[\rho \in (-0.95,-0.9,-0.8,-0.6,0.0) \quad \text{and} \quad H \in (0.5, 0.35, 0.2, 0.1).\]
We have also explored positive values of $\rho$, with results that are similar (up to symmetry) to the results for negative $\rho$ and therefore not reported here. For the call options we use a time-to-maturity of one year and set today's stock price to $S_0 = 1$. The strike price is varied over a grid of 
\[K \in (0.6,0.8,1.0,1.25,1.66).\]
For each parameter combination we simulate $10.000$ paths on a time grid of $1000$ steps. We evaluate the pathwise hedging error for both the variance-optimal and the Delta strategy and then calculate all relevant summary statistics, such as the mean-square hedging error or the relative error reduction.

\begin{figure}
\centering
\includegraphics[width=0.95\textwidth]{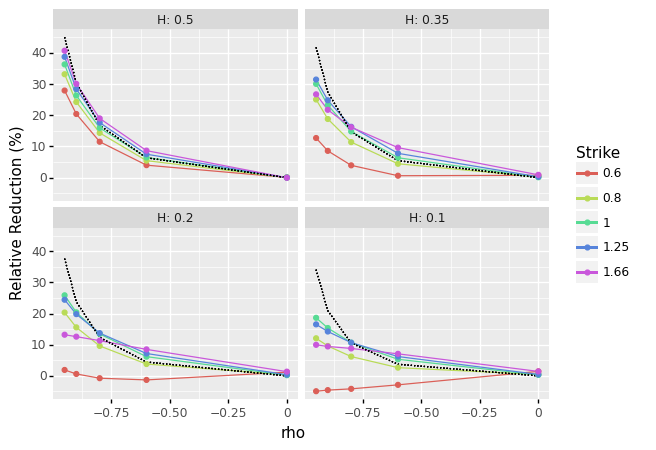}
\caption{\textbf{Relative Reduction in hedging error - simulation results.} The relative reduction in root-mean-squared hedging error for the variance-optimal strategy compared to the simple Delta hedging strategy is show in dependence on option strike price (colored lines), correlation $\rho$ and roughness parameter $H$. The black dotted line indicates the first-order approximation \eqref{eq:relred_rough}. \label{fig:relred}}
\end{figure}

\subsection{Observations}
We first focus on the empirical reduction in hedging error, which is shown in Figure~\ref{fig:relred}, together with the first-order approximation \eqref{eq:relred_rough}. For the SABR model ($H = 0.5$) and for the rough Bergomi model with large $H$, the first-order approximation matches well with the empirically observed error reduction. For smaller $H$, the approximation and the empirical error start to diverge, in particular for 
ITM options. Overall, the approximate variance-optimal hedge consistently leads to a reduction of the hedging error (in comparison to the Delta hedge), in particular for models with large $|\rho|$ and $H$ not too close to $0$. This is in line with our theoretical results from Sections~\ref{sec:SABR} and \ref{sec:rough}. Only for $H \le 0.2$ and for far ITM options a slightly negative reduction (i.e. an increase) of the empirical hedging error can be observed.\\

An even more detailed picture is painted by Figures~\ref{fig:densities_H05} and \ref{fig:densities_H02}, where we show estimated densities of the empirical pathwise hedging error for different combinations of $\rho$ and $K$. In particular for large $|\rho|$ it can be seen that the density of the variance-optimal strategies' error is narrower, more symmetric, and lighter-tailed then the density of the Delta strategies' error. The density of the Delta strategy error exhibits a skew that varies with strike $K$ and has heavier and asymmetric tails. However, already for the moderate values of $\rho = -0.6$ the effect is much diminished, in line with the theoretical results of Sections~\ref{sec:SABR} and \ref{sec:rough}, see also Figure~\ref{fig:relred_theory}. Comparing Figures~\ref{fig:densities_H05} and \ref{fig:densities_H02} it 
can also be seen that the described differences are more pronounces in the (non-rough) SABR case $H = 0.5$ than in the rough Bergomi case with $H = 0.2$. Again, this observation confirms our theoretical results.

\section{Conclusions}
In this article, we have derived analytic expressions for the variance-optimal hedging strategy and its mean-square error in the lognormal SABR and in the rough Bergomi model. For the SABR model, we find that the variance-optimal strategy coincides with Bartlett's Delta strategy from \cite{bartlett2006hedging}. Both theoretical results and simulations show that the variance-optimal strategy has lower hedging error than the implied Delta strategy, but this advantage only becomes substantial for strongly correlated models with large $|\rho|$. The relative efficiency of variance-optimal hedging is also affected by the roughness parameter $H$, with smaller $H$ (that is, increasing roughness) diminishing the advantage over the simple Delta hedge.\\
The general results of Section~\ref{sec:bg} on variance-optimal hedging in the `dynamic implied-volatilty' setting can likely be applied to other models beyond the lognormal SABR/rough Bergomi model, whenever a tractable expression for the evolution of implied volatility is available. In particular, in future work, we aim to generalize results to the SABR/rough Bergomi model with variable $\beta$, shedding further light on the stability of Bartlett's Delta with respect to variations in $\beta$, as discussed in \cite{hagan2017bartlett}. 

\begin{figure}
\centering
\includegraphics[width=0.95\textwidth]{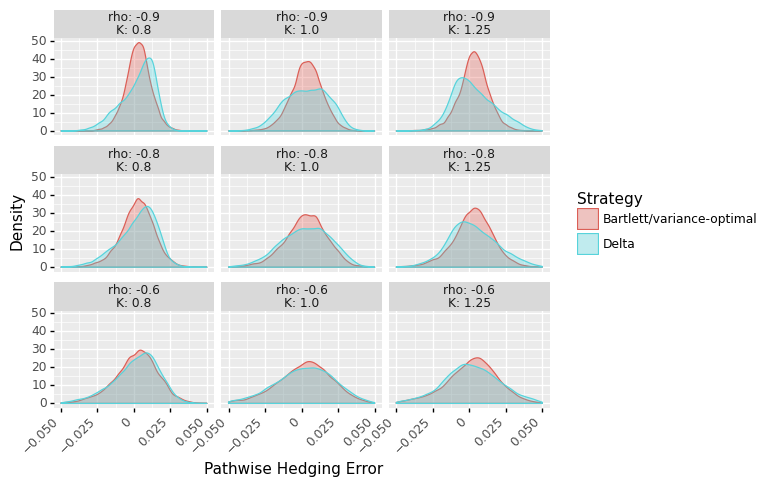} 
\caption{\textbf{Densities of pathwise hedging error in the SABR model.} 
This plot shows the simulated densities of the pathwise hedging error in the SABR model (Hurst parameter $H = 0.5$), comparing the Bartlett/variance-optimal strategy (red) and the Delta hedge (blue). Different combinations of correlation $\rho$ and option strike price $K$ are considered. \label{fig:densities_H05}}
\end{figure}

\begin{figure}
\centering
\includegraphics[width=0.95\textwidth]{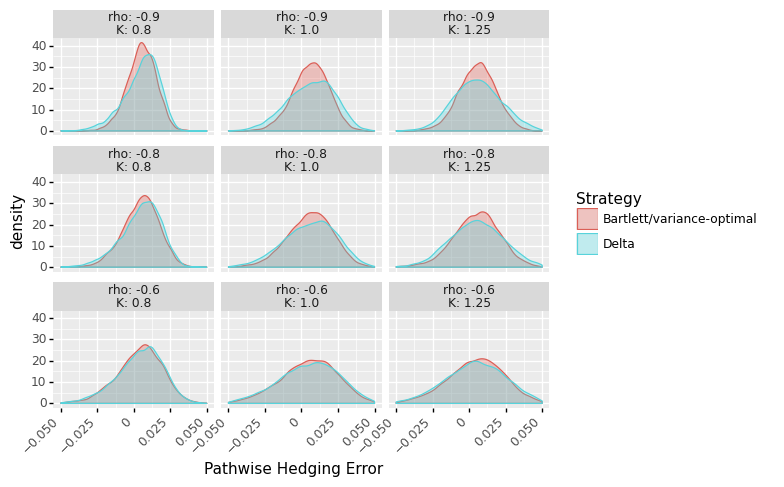} 
\caption{\textbf{Densities of pathwise hedging error in the rough Bergomi model ($H = 0.2$).} 
This plot shows the simulated densities of the pathwise hedging error in the rough Bergomi model (Hurst parameter $H = 0.2$), comparing the Bartlett/variance-optimal strategy (red) and the Delta hedge (blue). Different combinations of correlation $\rho$ and option strike price $K$ are considered.\label{fig:densities_H02} }
\end{figure}

\clearpage
%\section{Discussion}

%-------------------------------
%\bibliographystyle{imaiai}
\bibliographystyle{alpha}
\bibliography{references}
%-------------------------------

\pagebreak

\appendix

\section{Approximations for $f$, $F_1$ and $F_2$ in the rough Bergomi case}\label{app}
Following \cite{fukasawa2022rough} we set $G_H(y) = g(y)^2$ and report the closed-form expressions
\begin{align*}
G_{1/2}(y) &= 4 \left(\log \left(\frac{\sqrt{1 + \rho y + y^2/4} - \rho - y/2}{1- \rho}\right)\right)^2\\
G_{0}(y) &= \log\left(1 + 2\rho y + y^2\right) + \frac{2\rho}{\sqrt{1 - \rho^2}}
            \left(\arctan\left(\frac{\rho}{\sqrt{1 - \rho^2}}\right) - \arctan \left(\frac{y+\rho}{\sqrt{1 - \rho^2}}\right) \right)
\end{align*}
For $H \in (0,\tfrac{1}{2})$ we use the interpolation formula (cf.~\cite[Eq.(5.2)]{fukasawa2022rough})
\begin{equation}
G_H(y) = (2H + 1)^2 \left(c_0 G_0\left(\frac{y}{2H + 1}\right) + c_{1/2}  G_{1/2}\left(\frac{2y}{2H + 1}\right)\right),
\end{equation}
where
\[c_0 = \frac{3(1 - 2H)}{2H + 3}, \qquad c_{1/2} = \frac{2H}{2H + 3}.\]
For the derivatives, we obtain
\begin{align*}
G'_{1/2}(y) &= -8 \log \left(\frac{\sqrt{1 + \rho y + y^2/4} - \rho - y/2}{1- \rho}\right) / \sqrt{4\rho y + y^2 + 4}\\
G'_{0}(y) &= \frac{2y}{1 + 2\rho y + y^2}
\end{align*}
and 
\begin{equation}
G'_H(y) = (2H + 1)^2 \left(\frac{c_0}{2H+1} G'_0\left(\frac{y}{2H + 1}\right) + \frac{2 c_{1/2}}{2H+1} G'_{1/2}\left(\frac{2y}{2H + 1}\right)\right).
\end{equation}
Now, $f$, $F_1(y)$ and $F_2(y)$ can be calculated as
\begin{equation*}
f(y) = \frac{|y|}{\sqrt{G_H(y)}}, \qquad
F_1(y) = \sign{y} \cdot \frac{y^2}{2} \frac{G'_H(y)}{G_H(y)^{3/2}}, \qquad 
F_2(y) = \frac{2}{y} \left(F_1(y) - f(y)\right).
\end{equation*}
All functions are continuous at $y=0$ with values
\[f(0) = 1, \quad F_1(0) = 1, \quad F_2(0) = 0.\]

\end{document}